\documentclass[11pt]{amsart}

\newcommand{\jopinfo}{
	\author[Jop Bri\"{e}t]{Jop Bri\"et}
	\address{CWI, Science Park 123, 1098 XG Amsterdam, The Netherlands}
	\email{j.briet@cwi.nl}
}

\newcommand{\jopthanks}{
	\thanks{J.~B.\ was supported by a VENI grant from the Netherlands Organisation for Scientific Research (NWO)}
}

\usepackage{fullpage}
\usepackage{latexsym,amssymb,amsfonts,amsmath,mathrsfs,float}
\usepackage[font=small,labelfont=bf, width=.618\textwidth]{caption} 
\usepackage{xcolor}
\usepackage{tikz, graphics, enumerate}
\usepackage{hyperref}

  \renewcommand{\Pr}{\mbox{\rm Pr}}
  \newcommand{\Exp}{{\mathbb{E}}}
  \newcommand{\E}{\mathbb{E}}

  \newcommand{\R}{\mathbb{R}} 
  \newcommand{\N}{\mathbb{N}} 
  \newcommand{\F}{\mathbb{F}} 
  \newcommand{\bset}[1]{\{0,1\}^{#1}} 
  \DeclareMathOperator{\vspan}{Span} 
  \newcommand{\1}{\mathbf{1}}

  \newcommand{\st}{:\,} 

  \newcommand{\eps}{\varepsilon}
  \newcommand{\ip}[1]{\langle #1 \rangle}

  \newcommand{\poly}{\mbox{\rm poly}}

  \DeclareMathOperator{\spec}{Spec}
  \DeclareMathOperator{\enc}{\sf Enc}
  \DeclareMathOperator{\dec}{\sf Dec}

  \newcommand{\beq}{\begin{equation}}
  \newcommand{\eeq}{\end{equation}}
  \newcommand{\beqn}{\begin{equation*}}
  \newcommand{\eeqn}{\end{equation*}}
  \newcommand{\beqr}{\begin{eqnarray}}
  \newcommand{\eeqr}{\end{eqnarray}}
  \newcommand{\beqrn}{\begin{eqnarray*}}
  \newcommand{\eeqrn}{\end{eqnarray*}}
  \newcommand{\bmline}{\begin{multline}}
  \newcommand{\emline}{\end{multline}}
  \newcommand{\bmlinen}{\begin{multline*}}
  \newcommand{\emlinen}{\end{multline*}}

  \theoremstyle{plain}
  \newtheorem{theorem}{Theorem}[section]
  \newtheorem{lemma}[theorem]{Lemma}
  \newtheorem{proposition}[theorem]{Proposition}

  \theoremstyle{definition}
  \newtheorem{definition}[theorem]{Definition}
  
  \newtheorem{conjecture}[theorem]{Conjecture}
  
  \theoremstyle{remark}
  
  \renewenvironment{proof}[1][]{
    	\begin{trivlist}
     	\item[\hspace{\labelsep}{\em\noindent Proof#1:\/}]}
     	{{\hfill$\Box$}
    	\end{trivlist}
  }
  
  \makeatletter
  \newtheorem*{rep@theorem}{\rep@title}
  \newcommand{\newreptheorem}[2]{%
  \newenvironment{rep#1}[1]{%
  \def\rep@title{#2 \ref{##1}}%
  \begin{rep@theorem}}%
  {\end{rep@theorem}}}
  \makeatother

  \newreptheorem{lemma}{Lemma}
  \newreptheorem{theorem}{Theorem}
  \newreptheorem{corollary}{Corollary}
  \newreptheorem{proposition}{Proposition}
  \newreptheorem{conjecture}{Conjecture}

\begin{document}

\title[Freiman-Ruzsa Theorem in $\F_p^n$, NMCs and Testing]{Revisiting the Sanders-Freiman-Ruzsa Theorem in $\F_p^n$ and its Application to Non-malleable Codes}
\author{Divesh Aggarwal}
\address{School of Computer and Communication Sciences, EPFL, INF 230, Station 14, CH-1015, Lausanne, Switzerland}
\email{divesh.aggarwal@epfl.ch}
\jopinfo\jopthanks

\begin{abstract}
Non-malleable codes (NMCs) protect sensitive data against degrees of corruption that prohibit error detection, ensuring instead that a corrupted codeword decodes correctly or to something that bears little relation to the original message.
The \emph{split-state model}, in which codewords consist of two blocks, considers adversaries who tamper with either block arbitrarily but independently of the other.
The simplest construction in this model, due to Aggarwal, Dodis, and Lovett (STOC'14), was shown to give NMCs sending $k$-bit messages to $O(k^7)$-bit codewords.
It is conjectured, however, that the construction allows linear-length codewords.

Towards resolving this conjecture, we show that the construction allows for code-length~$O(k^5)$.
This is achieved by analysing a special case of
Sanders's \emph{Bogolyubov-Ruzsa theorem} for general Abelian groups.
Closely following the excellent exposition of this result for the group~$\F_2^n$ by Lovett, we expose its dependence on~$p$ for the group~$\F_p^n$, where $p$ is a prime. 
\end{abstract}

\maketitle

\section{Introduction}
\label{intro}

\subsection{Non-malleable codes}
\emph{Non-malleable codes} (NMCs) aim to protect data when it is subjected to the kind of corruption that renders reliable error correction and detection impossible.
The defining feature of such codes is that an adversary who tampers with a codeword will have little control over what it decodes to.
Despite having appeared only recently~\cite{DPW10},  these codes already emerged as a fundamental object at the intersection of coding theory and cryptography, such as in the construction of non-malleable commitment schemes~\cite{CGMPU15,GPR15} and non-malleable encryption schemes~\cite{CMTV15,CDTV16}. 
The study of non-malleable codes falls into a much larger cryptographic framework of providing counter-measures against various classes of tampering attacks. This work was pioneered by the early works of~\cite{ISW03,GLMMR04,IPSW06}, and has since led to many subsequent models (see~\cite{KKS11,LL12} for an extensive discussion of these models).

No code can protect against a completely unrestricted adversary.
For this reason, NMCs are only required to work for restricted families of ``tampering functions'' that an adversary may inflict.
An NMC limits an adversary's control over the decoded message by introducing randomness in the encoding procedure, whereby the encoding function randomly samples a codeword from a distribution that depends on the message.
More formally, for an alphabet~$\Gamma$ and a family of tampering functions~$\mathcal F$ mapping $\Gamma^n$ to itself, an NMC that protects against $\mathcal F$ consists of a randomized encoding function $\enc:\bset{k}\to\Gamma^n$, mapping messages to $\Gamma^n$-valued random variables, and a deterministic decoding function~$\dec:\Gamma^n\to \bset{k}\cup\{\perp\}$, where $\perp$ denotes error detection.
Roughly, the pair $(\enc, \dec)$ satisfies the following property.
For every $x\in \bset{n}$ and for every $f\in \mathcal F$, the (random) codeword $X = \enc(x)$ decodes correctly as $\dec(X) = x$, but the corrupted version $Y = f(X)$ either decodes to~$x$, or to a random variable $\dec(Y)$ whose distribution is close to a distribution~$\mathcal D_f$ depending on~$f$ but \emph{not} on~$x$.\footnote{
We refer to~\cite{DPW10} for a more formal definition.}
The main goal is to design NMCs for large classes of tampering functions while maximizing the \emph{rate} $k/(n\log |\Gamma|)$.

The class of tampering functions that has been studied most in the past literature arises in the so-called {\em split-state} model.
In this model, the codeword index-set~$[n]$ is partitioned into two roughly equally-sized sets $I_1,I_2 \subseteq[n]$ and the tampering functions consist of pairs $f = (f_1, f_2$), where $f_i: \Gamma^{I_i} \to \Gamma^{I_i}$ is arbitrary.
Codewords are then seen as being ``split'' into two states $X=(X_1,X_2)$, where $X_i\in\Gamma^{I_i}$, and  tampered codewords have the form $Y=(f_1(X_1),f_2(X_2))$.

Aggarwal, Dodis, and Lovett~\cite{ADL14} gave the first and by far the simplest construction in the split-state model.
For a prime number $p$ and positive integer $n$, their encoding function sends $\bset{k}$ into $\F_p^n\times \F_p^n$, giving split-state codewords of length $2n$ over the alphabet $\Gamma = \F_p$.
Based on an improved construction of a so-called \emph{affine-evasive set} due to~\cite{Agg15}, their proof shows that the construction has the desired properties when $\log p = \Omega(k)$ and $n = \Omega(\log^6 p)$, which translates to a rate of roughly $\log^{-6}k$.
However, the authors conjecture that their construction gives NMCs for constant~$n$ (independent of~$p$), giving constant-rate codes.
Although constant-rate split-state NMCs were later shown in~\cite{ADKO15}, trying to search for the best-possible parameters for the~\cite{ADL14} construction is interesting for the following two reasons.
First, the construction from~\cite{ADL14} is much simpler than the construction of~\cite{ADKO15}, which was obtained by adding several layers of encodings to an already complex construction of Chattopadhyay and Zuckerman~\cite{CZ14}. 
Second, though the rate of~\cite{ADKO15} is a constant, this constant is very small 
and given the number of layers used in the construction, it is unlikely that it can be improved significantly. 
In contrast, there is no obvious reason why the construction of~\cite{ADL14, Agg15} cannot yield codes of rate~1/20.

Towards determining the optimal parameters for the~\cite{ADL14} construction, we show that it still works when~$n = C\log^4 p$ for a sufficiently large constant~$C$, giving rate roughly $\log^{-4} k$.
To this end, we improve a key element of the security proof of the construction, namely the following striking property of the inner-product function.

\begin{theorem}[Aggarwal--Dodis--Lovett]\label{thm:ip}
There exist absolute constants $c,C\in (0,\infty)$ such that the following holds.
Let~$p$ be a prime, $n \geq C\log^6 p$ be an integer, $L,R$ be independent uniformly distributed random variables on~$\F_p^n$ and ${f,g:\F_p^n \to \F_p^n}$ be functions.
Then, there exist random variables $u, a, b$ on $\F_p$ such that~$u$ is uniformly distributed, $(a, b)$ is independent of~$u$, and the distributions of $(\ip{L,R},\langle f(L),g(R)\rangle)$ and $(u, au + b)$ have statistical distance at most~$2^{-c n^{1/6}}$.
\end{theorem}

The result roughly says that for any $f,g$, the random variable $\ip{f(L), g(R)}$ is correlated with some random variable of the form $au+b$.
The restriction on~$n$ imposed in the theorem is directly responsible for the restriction on the codeword length in the~\cite{ADL14} construction.
Improving this therefore implies higher-rate codes.
The proof of Theorem~\ref{thm:ip} relies crucially on a breakthrough result of Sanders~\cite{San12} in additive combinatorics, concerning sumsets in general Abelian groups (see below).
We improve Theorem~\ref{thm:ip} by simply exposing the dependence of Sanders's result on the magnitude of~$p$ when one restricts to the group~$\F_p^n$.
In particular, we obtain the following result (where all unspecified objects are as in Theorem~\ref{thm:ip}).

\begin{theorem}\label{thm:ip2}
There exist absolute constants $c,C\in (0,\infty)$ such that the following holds.
Let $n \geq C\log^4 p$ be an integer and $f,g:\F_p^n \to \F_p^n$ be functions. 
Then, there exist random variables $u, a, b$ on $\F_p$ such that~$u$ is uniformly distributed, $(a, b)$ is independent of~$u$, and the distributions of $(\ip{L,R},\langle f(L),g(R)\rangle)$ and $(u, au + b)$ have statistical distance at most~$2^{-c n^{1/4}}$.
\end{theorem}
 

\subsection{The Quasi-polynomial Fre\u{\i}man-Ruzsa Theorem}
For a finite Abelian group~$G$ and a subset~$A\subseteq G$, define the \emph{sum set} and \emph{difference set} to be $A + A = \{a + b\st a,b\in A\}$ and $A - A = \{a - b\st a,b\in A\}$, respectively.
The sizes of these sets are clearly bounded by~$|A|^2$,
but if~$A$ is a coset of a subgroup of~$G$, then these sizes are exactly~$|A|$.
Conversely, if $|A\pm A| = |A|$, then~$A$ must be a coset of a subgroup.
The identity $|A\pm A|/|A| = 1$ thus allows one to infer that~$A$ possesses a lot of structure.
One of the most important conjectures in additive combinatorics, the \emph{Polynomial Fre\u{\i}man-Ruzsa (PFR) Conjecture} (attributed to Marton in~\cite{Ruz99}), states that similar inferences can be made for sets in the group~$\F_2^n$ that satisfy $|A+A|/|A| \ll |A|$.


\begin{conjecture}[PFR Conjecture]\label{conj:PFR}
Let $A\subseteq \F_2^n$ be such that $|A-A| \leq K|A|$.
Then, there exists a set $B\subseteq A$ of size $|B| \geq |A|/C_1(K)$ that is contained in a coset of a subspace of size at most $C_2(K)|A|$,
where $C_1(K)$ and $C_2(K)$ are polynomial in~$K$.
\end{conjecture}



This conjecture is sometimes stated differently in the literature; see~\cite{Green:PFR} for the equivalence of five common formulations.
The above formulation is the one which appears most frequently in Theoretical Computer Science, where it found several important applications, such as in linearity testing~\cite{Sam07}, extractors~\cite{BZ11, AHL15}, error-correcting codes~\cite{BDL13} and communication complexity~\cite{BZL14}.
Major progress towards proving the PFR conjecture was made not long ago by Sanders~\cite{San12}, whose  result applies to general Abelian groups as opposed to just to~$\F_2^n$.
Recall that an Abelian group~$G$ has \emph{torsion}~$r$ if~$rg = 0$ for every~$g\in G$.\footnote{If the group operation is written multiplicatively, then the group is said to have exponent~$r$.}
The groups we care about here, namely~$\F_p^n$, thus have torsion~$p$.
For groups of bounded torsion, Sanders's result implies the following~\cite[Theorem~11.1]{San12}.



\begin{theorem}[Bogolyubov-Ruzsa Lemma for bounded torsion Abelian groups]\label{thm:sanders}
For every positive integer~$r$ there exists a $c(r)\in (0,\infty)$ such that the following holds.
Let~$G$ be an Abelian group of torsion~$r$, let~$A,B\subseteq G$ be such that $|A+B| \leq K\min\{|A|,|B|\}$.
Then, $(A-A)+(B-B)$ contains a subgroup~$|V|$ of size at least $|V|\geq |A+B|/2^{c\log^4 2K}$.
\end{theorem}

Standard arguments (see the proof of Lemma~\ref{lem:quasi4} in Section~\ref{sec:largesets}) show that Theorem~\ref{thm:sanders} implies the \emph{quasi-polynomial Fre\u{\i}man-Ruzsa Theorem}: the statement of Conjecture~\ref{conj:PFR} but with $C_1(K) = 2^{-C\log^4K}$ and $C_2(K) = K^c$ for absolute constants $C,c\in (0,\infty)$.
For NMCs, it was shown in~\cite{ADL14} that the following corollary of~\cite{San12} --- an~$\F_p$-analogue of the quasi-PFR Theorem~\cite[Lemma~10]{ADL14} --- implies Theorem~\ref{thm:ip}.

\begin{lemma}\label{lem:FpQPFR}
There exist absolute constants $C,c\in (0,\infty)$ such that the following holds.
Let $p$ be a prime, $n$ be a positive integer, $A\subseteq \F_p^n$ be such that $|A-A| \leq K|A|$.
Then, there exists a set $B\subseteq A$ of size $|B| \geq |A|/p^{C\log^6K}$ such that $\vspan(B)| \leq K^c |A|$.
\end{lemma}


Our improvement over Theorem~\ref{thm:ip}, Theorem~\ref{thm:ip2}, follows from the following variant of the above lemma, which we derive by exposing the dependence on the torsion~$p$ of~$\F_p^n$ in Theorem~\ref{thm:sanders}.

\begin{lemma}\label{lem:quasi4}
Let 
$A\subseteq \F_p^n$ 
be such that $|A-A| \leq K|A|$.
Then, there exists a set $B\subseteq A$ of size $|B| \geq |A|/p^{C\log^4(Kp)}$ such that $|\vspan(B)|\leq pK^c|A|$,
for absolute constants $C,c\in (0,\infty)$.
\end{lemma}

This improves Lemma~\ref{lem:FpQPFR} if $K = \poly(p)$. 
Since this is the case for the application to NMCs, the proof of Theorem~\ref{thm:ip} given in~\cite{ADL14}, but based on Lemma~\ref{lem:quasi4} instead of Lemma~\ref{lem:FpQPFR}, gives Theorem~\ref{thm:ip2}.

\subsection{Linearity testing}
One further application of Lemma~\ref{lem:quasi4} is to linearity testing.
The linearity test of Samoronidsky~\cite{Sam07} checks if a function $f:\F_p^n \mapsto \F_p^n$ is linear by picking $x, x' \in \F_p^n$ uniformly at random and accepting (declaring $f$ linear) if and only if $f(x - x') = f(x) - f(x')$. 
Based on Lemma~\ref{lem:FpQPFR} it was shown in~\cite{Sam07,ADL14} that the test rejects if $f$ is far from being linear. 
More precisely, it was shown that if $\Pr_{x, x' \in \F_p^n} [f(x-x') = f(x) - f(x')] \ge \eps$, then there exists a linear function $g:\F_p^n\to\F_p^n$ such that $\Pr[f(x) = g(x)] \ge \eps'$, where $\eps' = p^{-O(\log^6 (1/\eps))}$ for an absolute constant $c\in(0,\infty)$. 
Using Lemma~\ref{lem:quasi4} in the proof instead results in a bound of $\eps' = p^{-O(\log^4 (p/\eps))}$, which is an improvement if $\eps = p^{-\Omega(1)}$.

\section{Reduction to a Bogolyubov-Ruzsa Lemma for large sets}
\label{sec:largesets}

To avoid repeating the phrase \emph{``for absolute constants''} many times, in the sequel $C,C'$ will denote such constants whose value in $(0,\infty)$ may change from line to line.
For the proof of Lemma~\ref{lem:quasi4} we closely follow Lovett's excellent exposition~\cite{Lov15} of the result for the group~$\F_2^n$.
The first step is a standard reduction to the following special case of Theorem~\ref{thm:sanders}.

\begin{theorem}\label{thm:subspace-large}
Let~$A\subseteq\F_p^n$ be such that $|A| \geq p^{n-1}/K^{24}$ and $|A-A|\leq K|A|$. 
Then, there exists a subspace~$V\subseteq 2A - 2A$ of size~$|V| \geq p^{-C\log^4 (K \cdot p)}|A|$.
\end{theorem}

The reduction uses the fact that Theorem~\ref{thm:subspace-large} implies the following seemingly stronger version in which~$|A|$ is not bounded. 

\begin{theorem}\label{thm:subspace}
Let~$A\subseteq\F_p^n$ be such that~$|A-A|\leq K|A|$.
Then, there exists a subspace~$V\subseteq 2A-2A$ of size~$|V| \geq |A|/p^{C\log^4 (K \cdot p)}$.
\end{theorem}

Lemma~\ref{lem:quasi4} follows by combining the above theorem with the following result of Pl\"{u}nnecke~\cite{Plu69}.

\begin{lemma}[Pl\"{u}nnecke]\label{lem:plunnecke}
Let $G$ be an Abelian group and $ A\subseteq G$ be such that~$|A-A|\leq K|A|$. 
Then, for positive integers~$k,\ell$, we have~$|k A - \ell A|\leq K^{k+\ell}|A|$.
\end{lemma}


\begin{proof}[ of Lemma~\ref{lem:quasi4}]
Let~$V\subseteq 2A - 2A$ be a subspace as in Theorem~\ref{thm:subspace}.
Let $R\subseteq A$ be a maximal set such that no two elements from~$R$ belong to the same coset of~$V$.
By Lemma~\ref{lem:plunnecke},
$|R||V| = |R+V| = |A+V| \leq K^5|A|$.
It follows that $|R| \leq K^5p^{C\log^4 (K \cdot p)}$.
Let $B = A\cap (V + g)$ be such that $|B|$ is maximal.
Then, $|B| \geq |A|/|R| \geq |V|/K^5 \geq |A|/K^{5}p^{C\log^4 (K \cdot p)}$
and moreover,
$
|\vspan(B)|
\leq
|\vspan(V + g)|
\leq
p|V|
\leq
pK^4|A|$.\qedhere
\end{proof}

\begin{lemma}\label{lem:largesets}
Theorem~\ref{thm:subspace-large} implies Theorem~\ref{thm:subspace}.
\end{lemma}

The proof of Lemma~\ref{lem:largesets} uses the following definition and proposition.

\begin{definition}
Let~$A\subseteq\F_p^n$ be a subset and $t,m$ be positive integers and $\phi:\F_p^n\to\F_p^m$ be a linear map.
Then~$\phi$ is a {\em Freiman homomorphism} of order~$t$ of the set~$A$ if for any $k, \ell \in \N \cup \{0\}$, $k + \ell = t$, it maps distinct elements~$a,b\in kA - \ell A$ to distinct elements~$\phi(a),\phi(b)$ in~$\F_p^m$.
\end{definition}

\begin{proposition}\label{prop:freiman}
Let~$A\subseteq\F_p^n$ be a subset and let~$t$ be a positive integer.
Let~$m$ be the smallest integer such that there exists a Freiman homomorphism~$\phi:\F_p^n\to\F_p^m$ of order~$t$ of the set~$A$.
Then,~$\phi(\F_p \cdot (tA - tA)) = \F_p^m$.
\end{proposition}

\begin{proof}
Towards a proof by contradiction, assume that there exists an~$x\in \F_p^m$ such that for any $\alpha \in \F_p \setminus \{0\}$, $\alpha \cdot x\not\in\phi(tA - tA)$. We will show that the existence of such a vector contradicts the assumption that~$m$ is minimal. To this end, let~$\psi:\F_p^m\to\F_p^{m-1}$ be an arbitrary linear map whose kernel is the subspace $S$ of $\F_p^m$ spanned by $x$. We show that~$\psi\circ\phi:\F_p^n\to\F_p^{m-1}$ is also a Freiman homomorphism of order $t$ of~$A$, which implies our contradiction to the minimality of~$m$.

Let~$a,b$ be arbitrary distinct vectors in $kA-\ell A$ for some $k, \ell \in \N \cup \{0\}$, $k + \ell = t$. Moreover, since (by assumption)~$\phi$ is linear and in particular a Freiman homomorphism, and $S \cap \phi(tA - tA) = \{0\}$, we have $\phi(a) - \phi(b) = \phi(a-b) \not\in S$. Hence, since the composition of two linear operators is again linear, 
$
\psi\circ\phi(a) - \psi\circ\phi(b) =  \psi\big(\phi(a-b)\big)\ne 0$,
which is to say that $\psi\circ\phi$ maps distinct elements of~$kA-\ell A$ to distinct elements in~$\F_p^{m-1}$. In other words~$\psi\circ\phi$ is a Freiman homomorphism of order~$t$ of~$A$, giving the desired contradiction.
\end{proof}


\begin{proof}[ of Lemma~\ref{lem:largesets}]
Let~$A\subseteq\F_p^n$ be a set such that~$|A-A|\leq K|A|$ as in the statement of Theorem~\ref{thm:subspace}. We begin by observing that without loss of generality, we may assume that~$0\in A$. To see this, let~$a\in A$ be an arbitrary element and consider the set~$A' = A - a$, which certainly contains the origin. Also, $A' - A' = A-A$ and therefore~$|A'-A'| = |A-A|\leq K|A| = K|A'|$, but also~$2A' - 2A' = 2A-2A$. Hence, if there exists a subspace~$V\subseteq 2A'-2A'$ of the size promised by Theorem~\ref{thm:subspace} then~$V$ is also a subspace of~$2A-2A$.

Notice that since the assumption~$0\in A$ implies that~$\ell A\subseteq tA$ for every~$\ell\in[t]$, it follows that a Freiman homomorphism of order~$t$ of $A$ is also a Freiman homomorphism of order $\ell$ of~$A$.

To deduce the conclusion of Theorem~\ref{thm:subspace} from Theorem~\ref{thm:subspace-large} we use a Freiman homomorphism to embed~$A$ into a smaller space such that it occupies a large fraction of it. To this end, let~$\phi:\F_p^n\to\F_p^m$ be a Freiman homomorphism of order $12$ of~$A$ such that~$m$ is minimal. Let~$B = \phi(A) \subseteq\F_p^m$ be the embedding of~$A$ into~$\F_p^m$. Notice that since~$\phi$ is also a Freiman homohorphism of order $2$ of~$A$, we have~$|B-B| = |A-A| \leq K|A| = K|B|$. Moreover, by Lemma~\ref{lem:plunnecke} and Proposition~\ref{prop:freiman},
\begin{align*}
p^m &= |\F_p^m| = |\phi(\F_p \cdot (12 A - 12A))| = |\F_p(12\phi(A) - 12\phi(A))| \leq K^{24} \cdot p \cdot |B|,
\end{align*}
where the third equality uses linearity of~$\phi$.
Thus $|B|$ satisfies the conditions of Theorem~\ref{thm:subspace-large}.

Theorem~\ref{thm:subspace-large}	says that there exists a subspace~$V\subseteq 2B-2B$ of size $|V| \geq p^{-C\log^4 (K \cdot p)}|B|$. Since~$|B| = |A|$, the subspace~$|V|$ has the same size as the subspace promised to exist by Theorem~\ref{thm:subspace}. We now lift~$V$ into the larger space~$\F_p^n$ such that it becomes contained in~$2A-2A$. To this end, recall that~$\phi$ maps distinct elements of~$2A-2A$ into distinct elements of~$2B-2B$, since $2B-2B = \phi(2A-2A)$ and~$\phi$ is a Freiman homomorphism of order~4 of~$A$.
Hence for each element~$v\in V$ there is a unique element~$v'\in 2A-2A$ such that $\phi(v') = v$. Let~$V'\subseteq 2A-2A$ be the set of size~$|V'| = |V|$ such that~$\phi(V') = V$.

What is left is to show that~$V'$ is in fact a subspace of~$\F_p^n$. Clearly it suffices to show that if~$u',v'\in V'$ then $u'+v'\in V'$. Let $u = \phi(u')$ and $v = \phi(v')$.
Since~$V$ is a linear subspace, $w = u+v\in V$.
Let~$w'$ be the unique element in~$V'$ such that~$\phi(w') = w$. 
Since~$u',v,'w'\in 2A-2A$ we have $u'+v'- w'\in 6 A-6A$. 
As~$\phi$ is a Freiman homomorphims of order 12 of~$A$ and must map~$0\in 12A$ to~$0\in\F_p^m$, we have~$\phi(u'+v'-w') = u+v-w = 0$.
This implies~$u'+v'-w' = 0$ and thus~$u'+v' = w'\in V'$ as desired.
\end{proof}

%
\section{Proof of Theorem~\ref{thm:subspace-large}}

The proof of~Theorem~\ref{thm:subspace-large} is split into two parts.

\subsection{Part 1: A large gentle shifting set}

The first part uses the following standard notational conventions and definitions.
For a set $A \subseteq \F_p^n$, let $\mathbf{1}_A: \F_p^n \mapsto \{0,1\}$ be the indicator function of~$A$, and let $\rho_A(x) = (p^n/|A|) \mathbf{1}_A(x)$ be its density function. For $a \in \F_p^n$, we shorthand $\rho_{\{a\}}(x)$ by $\rho_{a}(x)$.
The convolution of two functions $f, g : \F_p^n \mapsto \R$ is defined by
$$
f * g(x) = \E_{y \in \F_p^n} [f(y) g(x-y)]
$$
and their inner product by 
$$
\langle f,g \rangle = \E_{x \in \F_p^n} [f(x) g(x)]. 
$$
Note the identities $\rho_A * f(x) = \E_{a \in A} [f(x-a)]$ and $\rho_x *f(a) = f(x - a)$.

The main result of this section is as follows.

\begin{lemma}\label{lem:shiftset}
Let~$A\subseteq\F_p^n$ be such that~$|A|\geq p^n/L$.
Then, for any~$\tau > 0$, there exists a set~$X\subseteq\F_p^n$ of size~$|X| \geq p^n/(2L^{C\tau^2\log^3 L})$ such that for every positive integer~$t \leq \tau\log L$ and every~$x\in tX$, we have
$
\Exp_{a,b\in A}[1_{A-A}(a-b-x)] \geq 0.9$.
\end{lemma}

The proof of the above lemma relies crucially on the following $\F_p^n$-version of a more general result of Croot and Sisask~\cite{Croot:2010}.

\begin{lemma}\label{lem:croot}
Let~$A\subseteq\F_p^n$ be such that~$|A| \geq p^n/L$ and let $f:\F_p^n\to [0,1]$ be a function. 
Then, for any $q\geq 1$ and $\eps\in (0,1)$, there exists a set~$X\subseteq\F_p^n$ of size~$|X| \geq p^n/(2L^{q/\eps^2})$ such that for every~$x\in X$, 
\beqn
\|\rho_x*\rho_A*f - \rho_A*f\|_q \leq C\eps.
\eeqn
\end{lemma}

\begin{proof}
Set~$\ell = q/\eps^2$.
We show that 
for each~$x\in\F_p^n$, at least half of the $\ell$-tuples~$(a_1,\dots,a_\ell)\in A^\ell$ satisfy
\beq\label{eq:empirical}
\big\|\rho_{(A+x)}*f - \frac{1}{\ell}\sum_{i=1}^\ell\rho_{a_i}*f\big\|_q \leq \frac{C\eps}{2}.
\eeq
To see why this suffices, for each~$x\in\F_p^n$ let $S_x\subseteq (\F_p^n)^\ell$ be the set of $\ell$-tuples in $A^\ell$ such that~\eqref{eq:empirical} holds.
We have~$|S_x| \geq |A|^\ell/2 \geq p^{n\ell}/(2 \cdot L^{\ell})$.
By the Pigeonhole Principle, there exists an $\ell$-tuple $(a_1,\dots,a_\ell)\in A^\ell$ that, for some set $X'\subseteq \F_p^n$ of size at least $p^n/(2 \cdot L^\ell)$,  belongs to~$S_x$ for each $x\in X'$.
By the triangle inequality, this implies that for every $x,x'\in X'$, we have
\begin{align}
\|\rho_{A+x}*f - \rho_{A+x'}*f\|_q
&\leq \big\|\rho_{(A+x)}*f - \frac{1}{\ell}\sum_{i=1}^\ell\rho_{a_i}*f\big\|_q
+\big\|\rho_{(A+x')}*f - \frac{1}{\ell}\sum_{i=1}^\ell\rho_{a_i}*f\big\|_q\nonumber\\
& \leq C\eps.\label{eq:shiftcroot}
\end{align}
Since the~$\ell_q$ norm is invariant under permutations of the coordinates, the left-hand side of~\eqref{eq:shiftcroot} equals
\beqn
\|\rho_{A+x}*f - \rho_{A+x'}*f\|_q = \|\rho_{A+x+x'}*f - \rho_{A}*f\|_q.
\eeqn
For any~$x\in X'$, the set~$X = X' + x$ then has the desired properties since~$|X| = |X'| \geq p^n/(2L^\ell)$.

We thus set out to prove that for every~$x\in\F_p^n$, inequality \eqref{eq:empirical} holds for at least half of~$A^\ell$.
Let~$a_1,\dots,a_\ell$ be independent uniformly distributed $A$-valued random variables. For each~$i\in[\ell]$ and~$x,y\in\F_p^n$ define the random variable $Y_i^x(y) = \rho_{A+x}*f(y) - \rho_{a_i}*f(y)$.
Notice that since the functions~$\rho_{A+x}*f$ and $\rho_{a_i}*f$ are~$[0,1]$-valued, we have~$|Y_i^x(y)|\leq 2$.
Hence, by definition of the $\ell_q$ norms, linearity of expectation and the definition of the~$Y_i^x(y)$ random variables,
\begin{align*}
\Exp\Big[\big\|\rho_{A+x}*f - \frac{1}{\ell}\sum_{i=1}^\ell \rho_{a_i}*f\big\|_q^q\Big]
&=\Exp\Big[\Exp_{y\in\F_p^n}\Big[\big|\frac{1}{\ell}\sum_{i=1}^\ell\big(\rho_{A+x}*f (y)-  \rho_{a_i}*f(y)\big)\big|^q\Big] \Big] \\
&=\Exp_{y\in\F_p^n}\Big[\Exp\Big[\big|\frac{1}{\ell}\big(Y_1^x(y) + \cdots + Y_\ell^x(y)\big)\big|^q \Big]\Big] \\
&\leq (C'q/\ell)^{q/2} \\ 
&=(\sqrt{C'}\eps)^q,
\end{align*}
where the last two lines follow from the Marcinkiewicz-Zygmund inequality~\cite{MZ37} and our choice of~$\ell$. 
Hence, by Markov's inequality,
\beq\label{eq:markov}
\Pr\Big[\big\| \rho_{A+x}*f - \frac{1}{\ell}\sum_{i=1}^\ell \rho_{a_i}*f\big\|_q \leq \frac{\sqrt{C}\eps}{2}\Big] \geq \frac{1}{2},
\eeq
showing that if we let~$C = \sqrt{C'}$, at least half of the~$\ell$-tuples in~$A^\ell$ satisfy~\eqref{eq:empirical} as required.
\end{proof}

\begin{proof}[ of Lemma~\ref{lem:shiftset}]
Let~$q\geq 1$ and $\eps\in(0,1)$ be parameters to be set later and let~$X\subseteq \F_p^n$ be a set as promised to exist by Lemma~\ref{lem:croot}.
Let~$t$ be a positive integer.
We begin by observing that for every~$x\in tX$,
\begin{align*}
\Pr_{a,b \in A} [a-b-x \in A - A]
&= \Exp_{a,b\in A}[\1_{A-A}(a-b-x)]\\ 
&= \langle\rho_x*\rho_A*\mathbf{1}_{A - A},\rho_A\rangle.
\end{align*}
Moreover, by H\"{o}lder's inequality,
\begin{align}
1 - \langle\rho_x*\rho_A*\1_{A-A},\rho_A\rangle
&= \langle\rho_A*\1_{A-A} - \rho_x*\rho_A*\1_{A-A},\rho_A\rangle\nonumber\\
&\leq \|\rho_A*\1_{A-A} - \rho_x*\rho_A*\1_{A-A}\|_q \|\rho_A\|_r,\label{eq:shiftnorms}
\end{align}
where~$r$ is defined by~$1/q + 1/r = 1$.
To lower bound the above expectations we upper bound the two norms in~\eqref{eq:shiftnorms}.
By our assumption on the size of~$A$, we have~$\|\rho_A\|_r = (p^n/|A|)^{(r-1)/r} \leq L^{1/q}$.
Let $x = x_1 + \cdots + x_t$, where $x_i \in X$ for $1 \le i \le t$. To bound the remaining norm we telescope the difference
\begin{multline*}
\rho_A*\1_{A-A} - \rho_{x_1 + \cdots + x_t}*\rho_A*\1_{A-A} =
\sum_{i=1}^t (\rho_{x_1+\cdots + x_{i-1}}*\rho_A*\1_{A-A} - \rho_{x_1+\cdots + x_i}*\rho_A*\1_{A-A}).
\end{multline*}
Then, by the triangle inequality, invariance of the~$\ell_q$ norm under permutations of the coordinates and Lemma~\ref{lem:croot} we get
\begin{align*}
\|\rho_A*\1_{A-A} - \rho_{x_1 + \cdots + x_t}*\rho_A*\1_{A-A}\|_q
&\leq 
\sum_{i=1}^t \|\rho_{x_1+\cdots + x_{i-1}}*\rho_A*\1_{A-A} - \rho_{x_1+\cdots + x_i}*\rho_A*\1_{A-A}\|_q\\
&=
\sum_{i=1}^t \|\rho_A*\1_{A-A} - \rho_{x_i}*\rho_A*\1_{A-A}\|_q \leq tC\eps.
\end{align*}

Hence, for every~$x\in tX$, we have
\begin{align*}
\Exp_{a,b}[1_{2A}(x+a+b)]
\geq 1 - \|\rho_A*\1_{A-A} -\rho_x*\rho_A*\1_{A-A}\|_q \|\rho_A\|_r
\geq 1 - tC\eps \cdot L^{1/q} \;.
\end{align*}
Set~$\eps = 1/(20 C t)$, $q = \log L$.
Then, for~$t \leq \tau\log L$ the above bound becomes
$
0.9
$
and the size of~$X$ is at least~$|X| \geq p^n/(2L^{q/\eps^2}) \geq p^n/(2L^{400 C^2 \tau^2\log^3 L})$.
\end{proof}

\subsection{Part 2: A subspace from the Fourier support}

Denote $\omega = e^{2\pi i/p}$.
For $X \subset \F_p^n$, $u\in\F_p^n$, and $\gamma\in[0,1]$, define 
$$
\widehat{X}(u) = \E_{x \in X} [\omega^{\langle u, x\rangle}]$$
and define
$$
\spec_\gamma(X) = \{u\in\F_p^n\st |\widehat X(u)|\geq \gamma\}.
$$
For $Y\subseteq\F_p^n$, let
$$
Y^\perp = \{v\in \F_p^n\st \langle u,v\rangle = 0,\quad \forall u\in Y\}.
$$

\begin{theorem}[Chang~\cite{Chang:2002}]\label{thm:chang}
Let~$X\subseteq\F_p^n$. Then
$$
\dim(\spec_\gamma(X)) \leq 8\gamma^{-2}\log(p^n/|X|).
$$
\end{theorem}


\begin{lemma}\label{lem:thespace}
Let$~t$ be a positive integer, $A,X\subseteq\F_p^n$.
Let $a,b\in A$, $x_1, \ldots, x_t \in X$ be independent uniformly distributed random variables and let $x = x_1 + \cdots + c_t$.
Assume that
\beq\label{eq:thespace1}
\Exp[\1_{A-A}(a-b-x)]  \geq 1-\eps.
\eeq
Set $V = \vspan(\spec_{1/2}(X))^{\perp}$.
Then, for an independent and uniformly distributed $v\in V$, we have
\beq
\Exp[\1_{A-A}(a-b-x+v)] \geq 1-\eps - \frac{1}{2^t}\frac{p^{n}}{|A|}.
\eeq
\end{lemma}

\begin{proof}
We will show that 
\beq\label{eq:eq75}
\Big| \Exp[\1_{A-A}(a-b-x) - \Exp[\1_{A-A}(a-b-x +v)]\Big| \leq \frac{p^{n}}{2^t|A|} \;.
\eeq
To this end, we begin by observing that by the Fourier Inversion Formula, the first and second expectation in~\eqref{eq:eq75} can be written as
\begin{align*}
\sum_{u\in \F_p^n}\widehat A(u)\widehat A(-u)  \widehat X(-u)^t\widehat \1_{A-A}(u)
\quad\quad\text{and}\quad\quad
\sum_{u\in \F_p^n}\widehat A(u)\widehat A(-u)\widehat X(-u)^t\widehat V(u)\widehat \1_{A-A}(u),
\end{align*}
respectively.

Since~$V$ is a subspace, $\widehat V(u) = 1$ if~$u\in V^\perp$ and $\widehat V(u) =0$ otherwise.
The left-hand side of~\eqref{eq:eq75} therefore becomes
\beqn
\Big|\sum_{u\not\in V^\perp}\widehat A(u) \widehat A(-u)\widehat X(-u)^t\widehat \1_{A-A}(u)\Big|.
\eeqn
By definition of~$V$ we have $V^\perp = \vspan(\spec_{1/2}(X))$. It follows that if~$u\not\in V^\perp$ then $|\widehat X(-u)| <1/2$. Also it is easy to see that~$|\widehat \1_{A-A}(u)|\leq 1$.
Hence,
\begin{align*}
\Big|\sum_{u\not\in V^\perp}\widehat A(u) \widehat A(-u)\widehat X(-u)^t\widehat \1_{A-A}(u)\Big|
&\leq \frac{1}{2^t} \sum_{u\not\in V^\perp}|\widehat A(u)\widehat A(-u)| \\
&\leq \frac{1}{2^t} \sum_{u\in\F_p^n}|\widehat A(u)| \cdot |\widehat A(-u)|  \\
&\leq \frac{1}{2^t} \cdot \frac{p^n}{|A|} \;,
\end{align*}
which gives the result.
\end{proof}

\begin{proof}[ of Theorem~\ref{thm:subspace-large}]
Recall that~$A$ has size at least $|A| \geq 2^n/L$ for~$L = K^{24} \cdot p$. 
Lemma~\ref{lem:shiftset} says that there exists a set~$X\subseteq\F_p^n$ of size~$|X| \geq p^n/(2L^{100 C \log ^3L})$ such that for any~$t \leq 10\log L$ and any $x \in tX$, we have
$
\Exp_{a,b\in A}[\1_{A-A}(a-b-x)] \geq 0.9$.
Theorem~\ref{thm:chang} and Lemma~\ref{lem:thespace} combined give us a subspace~$V$ of size at least~$|V| \geq p^n/(2^{32}p^{C'\log^4L})$ where~$C' = 3200 C$, such that the left-hand side of~\eqref{eq:thespace1} can be written as and bounded by

\begin{align}\label{eq:average}
\Exp\Big[\Exp_{v\in V}[\1_{A-A}(a-b-x+v)]\Big] \geq 0.8,
\end{align}
where the outer expectation is taken over uniformly distributed $a,b\in A$ and ${x_1,\ldots, x_t \in X}$.

We show that~$V\subseteq 2A-2A$, implying the result.
By the averaging principle, \eqref{eq:average} shows that there exist $a,b\in A$ and $x\in tX$ such that $80\%$ of the~$v\in V$ satisfy~$a-b-x+v\in A-A$.
This implies that $v\in A-A - a+b+x$.
We use the pigeon hole principle to show that this implies that~$V\subseteq 2A-2A$. To this end, fix a~$v\in V\backslash\{0\}$ and notice that there exists a complete matching of $|V|/2$ pairs~$(v_1,v_2)\in V\times V$ such that~$v =v_1 + v_2$ (take any~$v_1\in V$ and set~$v_2 = v-v_1$). This implies that their exists $v_1, v_2 \in A - A - a + b +x$ such that $v_1 + v_2 = v$, which in turn implies that $v \in 2A - 2A$.
\end{proof}

\bibliographystyle{alphaabbrv}
\bibliography{fuzzy}

\newcommand{\etalchar}[1]{$^{#1}$}
\begin{thebibliography}{BSLRZ14}
\expandafter\ifx\csname urlstyle\endcsname\relax
  \providecommand{\doi}[1]{doi:\discretionary{}{}{}#1}\else
  \providecommand{\doi}{doi:\discretionary{}{}{}\begingroup
  \urlstyle{rm}\Url}\fi

\bibitem[ADKO15]{ADKO15}
D.~Aggarwal, Y.~Dodis, T.~Kazana, and M.~Obremski.
\newblock Non-malleable reductions and applications.
\newblock In \emph{The 47th ACM Symposium on Theory of Computing (STOC)}. 2015.

\bibitem[ADL14]{ADL14}
D.~Aggarwal, Y.~Dodis, and S.~Lovett.
\newblock Non-malleable codes from additive combinatorics.
\newblock In \emph{STOC}. ACM, 2014.

\bibitem[Agg15]{Agg15}
D.~Aggarwal.
\newblock Affine-evasive sets modulo a prime.
\newblock \emph{Information Processing Letters}, 115(2):382--385, 2015.

\bibitem[AHL15]{AHL15}
D.~Aggarwal, K.~Hosseini, and S.~Lovett.
\newblock Affine-malleable extractors, spectrum doubling, and privacy
  amplification.
\newblock \emph{{IACR} Cryptology ePrint Archive}, 2015:1094, 2015.

\bibitem[BDL13]{BDL13}
A.~Bhowmick, Z.~Dvir, and S.~Lovett.
\newblock New bounds for matching vector families.
\newblock In \emph{Proceedings of the forty-fifth annual ACM symposium on
  Theory of computing}, pages 823--832. ACM, 2013.

\bibitem[BSLRZ14]{BZL14}
E.~Ben-Sasson, S.~Lovett, and N.~Ron-Zewi.
\newblock An additive combinatorics approach relating rank to communication
  complexity.
\newblock \emph{Journal of the ACM (JACM)}, 61(4):22, 2014.

\bibitem[CDTV16]{CDTV16}
S.~Coretti, Y.~Dodis, B.~Tackmann, and D.~Venturi.
\newblock Non-malleable encryption: Simpler, shorter, stronger.
\newblock In \emph{Theory of Cryptography - 13th International Conference,
  {TCC} 2016-A, Tel Aviv, Israel, January 10-13, 2016, Proceedings, Part {I}},
  pages 306--335. 2016.

\bibitem[CGM{\etalchar{+}}15]{CGMPU15}
N.~Chandran, V.~Goyal, P.~Mukherjee, O.~Pandey, and J.~Upadhyay.
\newblock Block-wise non-malleable codes.
\newblock \emph{{IACR} Cryptology ePrint Archive}, 2015:129, 2015.

\bibitem[Cha02]{Chang:2002}
M.-C. Chang.
\newblock {A polynomial bound in Fre\u{\i}man's theorem}.
\newblock \emph{Duke Mathematical Journal}, 113(3):399--420, 2002.

\bibitem[CMTV15]{CMTV15}
S.~Coretti, U.~Maurer, B.~Tackmann, and D.~Venturi.
\newblock From single-bit to multi-bit public-key encryption via non-malleable
  codes.
\newblock In \emph{Theory of Cryptography - 12th Theory of Cryptography
  Conference, {TCC} 2015, Warsaw, Poland, March 23-25, 2015, Proceedings, Part
  {I}}, pages 532--560. 2015.

\bibitem[CS10]{Croot:2010}
E.~Croot and O.~Sisask.
\newblock A probabilistic technique for finding almost-periods of convolutions.
\newblock \emph{Geometric and functional analysis}, 20(6):1367--1396, 2010.

\bibitem[CZ14]{CZ14}
E.~Chattopadhyay and D.~Zuckerman.
\newblock Non-malleable codes in the constant split-state model.
\newblock \emph{FOCS}, 2014.

\bibitem[DPW10]{DPW10}
S.~Dziembowski, K.~Pietrzak, and D.~Wichs.
\newblock Non-malleable codes.
\newblock In \emph{ICS}, pages 434--452. Tsinghua University Press, 2010.

\bibitem[GLM{\etalchar{+}}03]{GLMMR04}
R.~Gennaro, A.~Lysyanskaya, T.~Malkin, S.~Micali, and T.~Rabin.
\newblock {Algorithmic Tamper-Proof (ATP)} security: Theoretical foundations
  for security against hardware tampering.
\newblock In M.~Naor, editor, \emph{First Theory of Cryptography Conference ---
  TCC 2004}, volume 2951 of \emph{LNCS}, pages 258--277.
  Spring{\-}er-Ver{\-}lag, February~19--21 2003.

\bibitem[GPR15]{GPR15}
V.~Goyal, O.~Pandey, and S.~Richelson.
\newblock Textbook non-malleable commitments.
\newblock \emph{{IACR} Cryptology ePrint Archive}, 2015:1178, 2015.

\bibitem[Gre]{Green:PFR}
B.~J. Green.
\newblock Notes on the {P}olynomial {F}re\u{\i}man-{R}uzsa {C}onjecture.
\newblock Expository note. Available at
  http://people.maths.ox.ac.uk/greenbj/papers/PFR.pdf.

\bibitem[IPSW06]{IPSW06}
Y.~Ishai, M.~Prabhakaran, A.~Sahai, and D.~Wagner.
\newblock Private circuits {II}: Keeping secrets in tamperable circuits.
\newblock In S.~Vaudenay, editor, \emph{Advances in Cryptology---EUROCRYPT
  2006}, volume 4004 of \emph{LNCS}, pages 308--327. Spring{\-}er-Ver{\-}lag,
  2006.

\bibitem[ISW03]{ISW03}
Y.~Ishai, A.~Sahai, and D.~Wagner.
\newblock Private circuits: Securing hardware against probing attacks.
\newblock In D.~Boneh, editor, \emph{Advances in Cryptology---CRYPTO~2003},
  volume 2729 of \emph{LNCS}. Spring{\-}er-Ver{\-}lag, 2003.

\bibitem[KKS11]{KKS11}
Y.~T. Kalai, B.~Kanukurthi, and A.~Sahai.
\newblock Cryptography with tamperable and leaky memory.
\newblock In \emph{Advances in Cryptology--CRYPTO 2011}, pages 373--390.
  Springer, 2011.

\bibitem[LL12]{LL12}
F.-H. Liu and A.~Lysyanskaya.
\newblock Tamper and leakage resilience in the split-state model.
\newblock In \emph{Advances in Cryptology--CRYPTO 2012}, pages 517--532.
  Springer, 2012.

\bibitem[Lov15]{Lov15}
S.~Lovett.
\newblock \emph{An Exposition of Sanders' Quasi-Polynomial Freiman-Ruzsa
  Theorem}.
\newblock Number~6 in Graduate Surveys. Theory of Computing Library, 2015.

\bibitem[MZ37]{MZ37}
J.~Marcinkiewicz and A.~Zygmund.
\newblock Quelques th{\'e}oremes sur les fonctions ind{\'e}pendantes’.
\newblock \emph{Fund. Math}, 29:60--90, 1937.

\bibitem[Pl{\"u}69]{Plu69}
H.~Pl{\"u}nnecke.
\newblock Eigenschaften und absch{\"a}tzungen von wirkungsfunktionen, bmwf,
  1969.

\bibitem[Ruz99]{Ruz99}
I.~Ruzsa.
\newblock An analog of {F}re\u{\i}man's theorem in groups.
\newblock \emph{Ast{\'e}risque}, 258(199):323--326, 1999.

\bibitem[Sam07]{Sam07}
A.~Samorodnitsky.
\newblock Low-degree tests at large distances.
\newblock In \emph{ACM symposium on Theory of computing}, pages 506--515. ACM,
  2007.

\bibitem[San12]{San12}
T.~Sanders.
\newblock On the {B}ogolyubov--{R}uzsa lemma.
\newblock \emph{Analysis \& PDE}, 5(3):627--655, 2012.

\bibitem[ZBS11]{BZ11}
N.~Zewi and E.~Ben-Sasson.
\newblock From affine to two-source extractors via approximate duality.
\newblock In \emph{Proceedings of the forty-third annual ACM symposium on
  Theory of computing}, pages 177--186. ACM, 2011.

\end{thebibliography}

\end{document}